\documentclass[10pt,journal]{IEEEtran}
\IEEEoverridecommandlockouts

\usepackage{cite}
\usepackage{amsmath}
\usepackage{amsthm}
\usepackage{amssymb}

\usepackage{graphicx}
\usepackage{textcomp}
\usepackage{xcolor}

\usepackage{algorithm}
\usepackage[noend]{algpseudocode}
\makeatletter
\def\BState{\State\hskip-\ALG@thistlm}
\makeatother
\newtheorem{proposition}{Proposition}
\newtheorem{definition}{Definition}
\newtheorem{remark}{Remark}
\newtheorem{assumption}{Assumption}

\usepackage{matlab-prettifier}
\usepackage{color}
\usepackage{mathrsfs}
\usepackage[shortlabels]{enumitem}
\makeatletter
\def\BState{\State\hskip-\ALG@thistlm}
\makeatother
\def\BibTeX{{\rm B\kern-.05em{\sc i\kern-.025em b}\kern-.08em
    T\kern-.1667em\lower.7ex\hbox{E}\kern-.125emX}}
\usepackage{subcaption}

\begin{document}

\title{{\bf\Large A Game-Theoretic Analysis of Auditing Differentially Private Algorithms with Epistemically Disparate Herd}
}

\author{Ya-Ting Yang, Tao Zhang, and Quanyan Zhu
\thanks{The Authors are with the Department of Electrical and Computer Engineering, New York University, Brooklyn, NY, 11201, USA; E-mail: {\tt\small \{yy4348, tz636, qz494\}@nyu.edu}.}% 
}

\maketitle

\begin{abstract}
Privacy-preserving AI algorithms are widely adopted in various domains, but the lack of transparency might pose accountability issues. While auditing algorithms can address this issue, machine-based audit approaches are often costly and time-consuming. Herd audit, on the other hand, offers an alternative solution by harnessing collective intelligence. Nevertheless, the presence of epistemic disparity among auditors, resulting in varying levels of expertise and access to knowledge, may impact audit performance. An effective herd audit will establish a credible accountability threat for algorithm developers, incentivizing them to uphold their claims. In this study, our objective is to develop a systematic framework that examines the impact of herd audits on algorithm developers using the Stackelberg game approach. The optimal strategy for auditors emphasizes the importance of easy access to relevant information, as it increases the auditors' confidence in the audit process. Similarly, the optimal choice for developers suggests that herd audit is viable when auditors face lower costs in acquiring knowledge. By enhancing transparency and accountability, herd audit contributes to the responsible development of privacy-preserving algorithms.
\end{abstract}
\section{Introduction}
AI and algorithmic decision-making have become pervasive in both business and society. However, when algorithms are treated as ``black boxes'' and their inner workings remain undisclosed, it becomes difficult to ensure that they perform as intended and adhere to necessary standards \cite{Why_audit}. One specific category of algorithms that exemplifies this challenge is privacy-preserving algorithms \cite{DP_survey}. For instance, platforms like Facebook Ad Recommendation Systems, Google SQL, and Safari have integrated differential privacy into their products to provide privacy protection. Nevertheless, verifying such claims can be arduous and intricate, for example, see \cite{VDP_detection,dp_finder,num_dp}.

%The adoption of privacy-preserving algorithms is a positive step toward addressing privacy concerns in AI applications. However, the lack of transparency and the inability to scrutinize these algorithms can cast doubt on the effectiveness of the privacy measures. Greater efforts need to be made to establish independent auditing mechanisms or standardized practices that allow for the verification and validation of privacy claims made by algorithmic systems. By promoting transparency and accountability, we ensure that privacy-preserving algorithms genuinely deliver on their promises and provide the necessary protection for users' privacy.

\noindent\textbf{Herd Audit:} Auditing algorithms \cite{audit_survey}, \cite{audit_transparent} play a crucial role in tackling this challenge. However, traditional machine-based audit methods like direct scraping, sock puppet, and carrier puppet often necessitate the development of custom computer programs to gather data. Not only can these approaches be expensive, but they also consume a significant amount of time.
A cost-effective alternative approach to auditing involves leveraging citizen science and crowd-sourcing principles to establish a democratic audit process that engages a diverse population of end users \cite{collective_intelligence}. This concept gives rise to \textit{herd-audit} (or group-audit) approaches.
%
%By harnessing the power of collective intelligence \cite{collective_intelligence}, herd-audit approaches tap into the knowledge, experiences, and perspectives of a wide range of individuals. This distributed effort can result in a more comprehensive and diverse audit of algorithms. Implementing herd-audit approaches not only addresses the limitations of traditional methods but also promotes transparency and inclusivity. It allows a broader segment of the population to actively participate in holding algorithmic systems accountable. 
By empowering end users as auditors, we can foster a more democratic approach to algorithmic auditing while minimizing costs and time investments.

\begin{figure}[ht]
    \centering \vspace{-3mm}
    \includegraphics[width=1.75in]{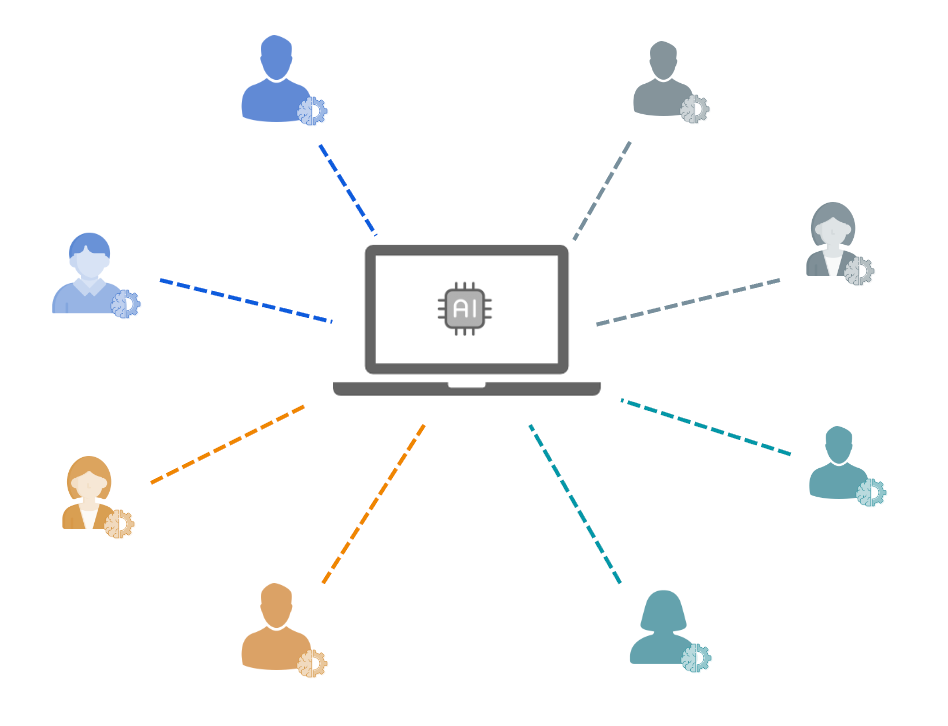}\vspace{-3mm}
    \caption{A herd of diverse end-users act as auditors to inspect the AI algorithm used in the developed product.}
    \label{fig:bigpicture}
    \vspace{-5mm}
\end{figure}

\noindent\textbf{Epistemic Disparity:} One significant challenge in implementing herd-audit approaches is the presence of epistemic disparity \cite{eps_injustice,eps_science}. Not all users possess the same level of expertise or information required to conduct comprehensive audits of algorithms. %There is a wide distribution of knowledge and variations in users' access to relevant information. These variations stem from differences in cognitive and reasoning capabilities among users. 
%A user-auditor who approaches the task with meticulous scrutiny is more likely to arrive at accurate outcomes during the evaluation process. Conversely, 
A user-auditor with limited cognitive resources may inadvertently provide opportunities for algorithm developers to evade their responsibility. %Addressing this challenge is essential to establish a reliable and effective approach that can elicit responsible algorithms through herd audit.
To some extent, incorporating audit into the algorithm design process itself establishes an accountability mechanism for developers. %When algorithms are transparent and capable of herd-audit, any deviations from the claimed performance are readily evident. For example, in relation to privacy protection, if an algorithm designer asserts the inclusion of differential privacy and the algorithm undergoes herd audit, any inconsistencies between the claimed privacy assurances and the actual performance will be revealed. 
This accountability mechanism acts as an incentive for algorithm developers to uphold their claims and create responsible algorithms.

\noindent\textbf{Game-Theoretic Framework:}
To design an effective herd-audit mechanism, this work aims to develop a comprehensive system framework that investigates the influence of herd-audit on algorithm developers. %as illustrated in Fig. \ref{fig:herd_audit}.
%One of the primary aims of this framework is to gain insights into the behavior and motivations of developers when subjected to herd audit. 
To accomplish the goal, the system framework adopts a Stackelberg game approach \cite{manshaei2013game,fang2021introduction}. In this approach, the developer assumes the role of the leader and determines the desired level of performance for differential privacy. The followers, comprising idiosyncratic end-users or auditors, are selected from a user population characterized by varying levels of epistemic capabilities. The proposed framework assumes that algorithms and their associated guarantees are clearly communicated to the end-users through a privacy protection agreement. This leader-and-follower structure allows us to analyze the optimal strategies employed by both the developer and the auditors, providing insights into the potential noncompliant behaviors of developers. 
%Furthermore, it helps in understanding how to incentivize developers to create responsible algorithms and explore the accountability mechanisms.

In order to capture the epistemic disparity experienced by end-users (auditors), this work employs a rational inattention model \cite{RI_discrete,RI_info_cost}, which takes into account the costs associated with accessing information during the decision-making process. 
%It acknowledges that end-users have limited cognitive resources and are unable to fully attend to or process all available information. One notable advantage of this model is its incorporation of the concept of mental effort, which provides a high-level abstraction of cognitive processes. 
%This abstraction facilitates the characterization of cognitive processes at a population level by aggregating diverse behaviors. 
%The rational inattention model has played an important role recently in explaining various economic phenomena involving individuals with cognitive constraints, such as consumer behavior and investment decisions.
We analyze the epistemic disparity among auditors, characterized by the epistemic factor, which measures the difficulty of accessing information. We find that auditors with lower epistemic factors exhibit higher audit confidence, indicating a better audit performance. Furthermore, our investigation reveals that a herd audit is a viable approach when auditors face lower costs in accessing information. In such circumstances, the algorithm developer is less likely to deviate significantly from their claims. Our findings highlight the importance of reducing epistemic injustice as well as lowering information costs to enhance the effectiveness of herd audits. By doing so, we can foster a more reliable and accountable environment for the development of algorithms.

%\noindent\textbf{Paper Organization:} The rest of this paper is organized as follows. Section 2 provides an overview of recent research on algorithm auditing and explores related game-theoretic approaches. Section 3 introduces the rational inattention model, which captures the reasoning process of herd auditors with varying levels of epistemic disparity. In Section 4, a Stackelberg herd audit game is proposed for a specific class of privacy-preserving algorithms. Section 5 discusses the equilibrium solutions that characterize the behaviors of both auditors and developers. Finally, Section 6 concludes the paper, summarizing the key findings and implications.

%%%%%%%%%%%%%%%%%%%%%%%%%
%\section{Related Work}
\noindent\textbf{Related Works}
Algorithm auditing refers to the process of evaluating the algorithms used in systems or applications to ensure they are fair, transparent, unbiased, and comply with ethical standards \cite{audit_survey}. In differential privacy, several machine-based verification methods have been proposed \cite{VDP_detection,dp_finder,num_dp}. 
%These machine auditors are capable of acquiring a large number of samples and utilizing the law of large numbers to estimate the probability of outcomes and verify the algorithms.
While there has been a rich literature on citizen science and its applications in crowdsensing \cite{crowdsensing}, crowdsourcing \cite{crowdsourcing}, and crowd defense \cite{crowd_defence}, herd audit is a concept in its infancy. It reduces auditing costs and poses a threat to developers, as public perception \cite{frye2021technology} can be influenced by the audit results.

The disparity in the capability of herd behaviors has been studied in collective intelligence \cite{yu2016mitigating,collective_intelligence,comeig2020rational,eickhoff2018cognitive}. The literature has examined the performance \cite{morris2012priming}, reliability \cite{karger2014budget}, and trustworthiness \cite{wang2016toward} of participants engaged in outsourced tasks. Processes such as risk and reputation management \cite{allahbakhsh2012reputation,yu2020crowdr} have been utilized to understand the differences among participants. 
%Notably, many studies have placed emphasis on the careful selection of participants to effectively achieve the goals of the task at hand. This body of work primarily focuses on understanding the impact of cognitive variabilities in herd audits, as participants exhibit diverse cognitive abilities in assessing whether algorithms perform as claimed.
%
Numerous studies have focused on different cognitive behaviors in humans, including cognitive-behavioral theory \cite{Cognitive-behavioral,why_cognitive,huang2023cognitive} which elucidates how thoughts, beliefs, and cognitive processes shape behavior, and 
%This theory underscores the significance of cognitive processes in interpreting and responding to environmental stimuli, while highlighting the potential for modifying these processes to achieve desired behavioral outcomes. 
the theory of mind \cite{anderson2004integrated} that attributes mental states such as beliefs and emotions to predict individuals' behavior. 
In our work, we employ the concept of rational inattention, as studied in \cite{Sims_2003}, which provides a framework that analyzes how decision-makers acquire information while considering associated costs, enabling investigations into cognitive impacts on audit decisions.
%Utilizing this framework offers several advantages, including a high-level abstraction of mental effort, which encompasses attention, perception, memory, and problem-solving in cognitive functioning. Additionally, it facilitates the exploration of the interplay between cognition and decision-making processes, enabling a focused investigation into the cognitive impact on audit decisions.

A game-theoretic approach is commonly employed to capture the threat posed by followers in dynamic games, such as ultimatum games \cite{rajtmajer2017ultimatum}, Stackelberg games \cite{casorran2019study}, bargaining games \cite{guerrero2018solving}, as well as contract  \cite{chen2016optimal,zhang2019mathtt} and incentive mechanisms designs \cite{zhu2012guidex,huang2021duplicity}. 
%In this study, we adopt a Stackelberg game framework to evaluate the dependability of herd auditors, who are modeled as randomly sampled idiosyncratic individuals, and to assess the opportunities available to the product designer to evade compliance. 
Recently, there has been increased interest in the investigation of evasion behaviors \cite{deception}. This includes exploiting evasion-aware detection methods \cite{evasion_detectors} and developing evaders for subsequent tests of collaborative cognition-assisted detector \cite{early_detection}. 

%Additionally, the literature on cyber deception has explored evasion strategies to gain a better understanding of attackers' stealthiness \cite{deception}.

%This approach allows us to draw insights and develop strategies to enhance the reliability and robustness of auditing systems in the face of potential evasion attempts.

%%%%%%%%%%%%%%%%%%%%%%%%%
\section{Herd Auditors with Epistemic Disparity}
In the context of herd-auditing an algorithm, the auditor is uncertain about the true state $\omega \in \Omega=\{g, b\}$, where $g$ indicates the null hypothesis, implying that the algorithm is consistent with the claim, while $b$ is for the alternative hypothesis, meaning that the algorithm does not comply. The prior belief of state $\omega$ can be denoted as $\mu(\omega)$, implying the auditor's uncertainty in the algorithm's compliance.

In order to reduce the uncertainty, the auditor can obtain information $s$ about the state according to the information-obtaining strategy $d(s|\omega)$. More specifically, $s$ can be viewed as the outcome of the algorithm, and $d(s|\omega)$ indicates how the auditor accesses (obtains) it. The information $s$ together with the obtaining strategy leads to a posterior belief of the state $\mu(\omega|s)=\frac{\mu(\omega)d(s|\omega)}{\sum_{\omega}\mu(\omega)d(s|\omega)}$.

Based on the information $s$ (correspondingly, the posterior belief $\mu(\omega|s)$), the auditor can select an element from a finite action set $a \in \mathcal{A}=\{T, F\}$, where $T$ means reporting algorithm compliance, while $F$ indicates reporting non-compliance. The decision rule  $\delta: \mathcal{S} \mapsto \mathcal{A}$ aims to maximize the expected utility of $u (\omega, a)$, where $u: \Omega \times \mathcal{A} \mapsto \mathbb{R}$ is the utility of choosing action $a$ when the state is $\omega$. 

\begin{figure}
    \centering \vspace{-3mm}
    \includegraphics[width=3in]{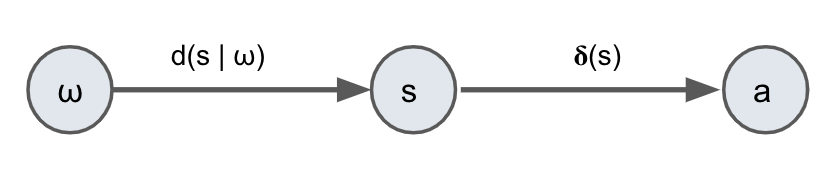}\vspace{-3mm}
    \caption{An illustration of how the auditor performs audit. The auditor acquires information $s$ about the unknown state $\omega$ using strategy $d$ and then makes a decision $a$ using strategy $\delta$.} \vspace{-3mm}
    \label{fig:stage}
\end{figure}

However, the acquisition of information can incur costs, which can be viewed as the discrepancy between the prior belief $\mu(\omega)$ and the posterior belief $\mu(\omega|s)$ regarding the state $\omega$. In conventional rational inattention research, a common method to model the cost is through the lens of Shannon mutual information. Furthermore, due to variations in epistemic disparities, the cost incurred for accessing information (i.e., reduction in uncertainty) differs among auditors. To account for this, we introduce the concept of an \textit{epistemic factor} for each auditor, denoted as $\lambda$, which quantifies the differences in the cost experienced by different auditors when reducing the same amount of uncertainty. The larger value of $\lambda$ implies harder access to relevant information, as the cost for the same amount of uncertainty reduction becomes higher.
To this end, the auditor's objective becomes 
\begin{equation}
    \max_{d, \delta} \mathbb{E}[u(\omega, a)] - \lambda I(\omega;s),
\label{eq:1}
\end{equation} where the expected utility is given by
\begin{equation}
    \mathbb{E}[u(\omega, a)]=\sum_{\omega}\sum_{a}\mu(\omega)u(\omega, a)\sum_{s:\delta(s)=a}d(s|\omega),
\end{equation} and the information cost is expressed as
\begin{equation}
    I(\omega;s)=\sum_{\omega}\sum_{s} d(s| \omega) \mu(\omega) \ln \frac{d(s | \omega)}{\sum_{\omega} d(s | \omega) \mu(\omega)}.
\end{equation}

\subsection{Bayes hypothesis testing as the auditor's decision rule}

Conventionally, Bayes hypothesis testing deals with the optimization problem
\begin{equation}
    \max_{\delta} \mathbb{E}[u(\omega, a)]=\sum_{\omega}\sum_{a}\mu(\omega)u(\omega, a)\sum_{\delta(s)=a}d(s|\omega)
\end{equation} with given distributions for both hypotheses $d(s|g)$ and $d(s|b)$ during decision-making, which coincides with the first term in the auditor's objective \eqref{eq:1}. According to detailed derivation in Appendix \ref{appendix: delta(s)}, the optimal decision rule can be written as
\begin{equation}
    \delta^*(s)=
    \begin{cases}
    T, \ &\frac{\mu(b)d(s|b)}{\mu(g)d(s|g)} < \frac{u(g, T)-u(g, F)}{u(b, F)-u(b, T)}, \\
    F, \ &\frac{\mu(b)d(s|b)}{\mu(g)d(s|g)} > \frac{u(g, T)-u(g, F)}{u(b, F)-u(b, T)}, \\
    \{T, F\}, &\frac{\mu(b)d(s|b)}{\mu(g)d(s|g)} = \frac{u(g, T)-u(g, F)}{u(b, F)-u(b, T)},
    \end{cases}
\end{equation} which leads us to a threshold decision rule and can be viewed as making a decision based on the posteriors. We represent the optimal decision rule with given $d(s|g)$ and $d(s|b)$ as $\delta^*_d(s)$, and denote the information set partitioned by $\delta^*_d(s)$ as 
\begin{equation}
    \begin{cases}
        S_{d, T}=\{s: \delta^*_d(s)=T\}, \\
        S_{d, F}=\{s: \delta^*_d(s)=F\}.
    \end{cases}
\end{equation}

\subsection{Auditor's choice of the information strategy} \label{sec:choice_info}

With the optimal decision rule $\delta^*_d$, the auditor's objective:
\begin{equation}
    \begin{aligned}
        \max_{d, \delta} \ &\mathbb{E}[u(\omega, a)] - \lambda I(\omega;s), \ \text{with} \  \delta=\delta^*_d, \\
    \end{aligned}
\end{equation} which leads to the constrained optimization problem
\begin{equation}
\begin{aligned}
    \max_{d} \ &\sum_{\omega}\sum_{a}\mu(\omega)u(\omega, a)\sum_{s:\delta^*_d(s)=a}d(s|\omega)\\
    &- \lambda \sum_{\omega}\sum_{s} d(s| \omega) \mu(\omega) \ln \frac{d(s | \omega)}{\sum_{\omega} d(s | \omega) \mu(\omega)}, \\
    \text{s.t.} \ & \sum_{s} d(s|\omega)= 1, d(s|\omega) \geq 0, \forall s \in \mathcal{S}, \forall \omega \in \Omega.
\end{aligned}
\label{eq:max_d}
\end{equation} 
With detailed derivations in Appendix \ref{appendix: d(s|w)} we arrive at:
\begin{align}
    &d(s|g)=\label{eq:d_g}
    \begin{cases}
        \frac{v(s)\exp(\frac{u(g, T)}{\lambda})}{y'(g)}, \ s \in S_{d, T},\\
        \frac{v(s)\exp(\frac{u(g, F)}{\lambda})}{y'(g)}, \ s \in S_{d, F},
    \end{cases} \\
    &d(s|b)=\label{eq:d_b}
    \begin{cases}
        \frac{v(s)\exp(\frac{u(b, T)}{\lambda})}{y'(b)}, \ s \in S_{d, T},\\
        \frac{v(s)\exp(\frac{u(b, F)}{\lambda})}{y'(b)}, \ s \in S_{d, F}.
    \end{cases}
\end{align}
The corresponding posterior belief $\mu(g|s)=\frac{\mu(g)d(s|g)}{\sum_{\omega}\mu(\omega)d(s|\omega)}=\frac{\mu(g)d(s|g)}{v(s)}$ can then be written as 
\begin{align}
    &
    \mu(g|s)=
    \begin{cases}
        \frac{\mu(g)\exp(u(g, T)/\lambda)}{y'(g)}, \ s \in S_{d, T},\\
        \frac{\mu(g)\exp(u(g, F)/\lambda)}{y'(g)}, \ s \in S_{d, F} \label{eq:post_g},
    \end{cases}
\end{align}
Note that the $s \in S_{d, T}$ case can be viewed as the posterior belief $\mu(g|s)$ given $s$ that  results in an action $a=T$ (i.e., $\mu(g|s)=\mu(g|T), \ s \in S_{d, T}$), while the $s \in S_{d, F}$ case can be viewed as the posterior belief $\mu(g|s)$ given $s$ that  results in an action $a=F$ (i.e., $\mu(g|s)=\mu(g|F), \ s \in S_{d, F}$). A similar expression can be found for $\mu(b|s)$.

\begin{align}
    &\mu(b|s)=
    \begin{cases}
        \frac{\mu(b)\exp(u(b, T)/\lambda)}{y'(b)}, \ s \in S_{d, T},\\
        \frac{\mu(b)\exp(u(b, F)/\lambda)}{y'(b)}, \ s \in S_{d, F} \label{eq:post_b},
    \end{cases}
\end{align} where $y'(g)$ and $y'(b)$ are corresponding normalization terms.

\begin{remark}
    For an auditor with epistemic factor $\lambda$, the information-obtaining strategy represented by the conditional probability $d(s|\omega)$ is chosen if its resulting posterior belief $\mu(\omega|s)$ maximizes the value of $\mathbb{E}[u(\omega, a)] - \lambda I(\omega;s)$. 
    %where $u(\omega, a)$ is the utility and $I(\omega;s)$ is the mutual information between $\omega$ and $s$ that measures the expected reduction of uncertainty from the prior to posterior belief caused by $d(s|\omega)$.
\label{remark:d_and_mu}
\end{remark}

The $\mu(g|s), \forall s \in S_{d, T}$, can also be interpreted as the \textit{audit confidence} for making the decision $a=T$ when observing the information $s$. Since $u(g, T) > u(g, F)$, it is evident that auditors with a smaller epistemic factor $\lambda$  have higher confidence in the audit process. This implies that auditors who can easily access relative information are more likely to perform better in the audit.

\section{Stackelberg Herd Audit Game}
To examine the impact of herd audit on the developer's incentive to behave irresponsibly, we formulate the interplay between the herd auditor (she) and the algorithm developer (he) as a Stackelberg herd audit game, depicted in Fig. \ref{fig:herd_audit}.

\subsection{Connection to differential privacy}
We begin with the definition of $\epsilon$-differential privacy. 
\begin{definition}[$\epsilon$-DP]
    A (randomized) mechanism $M:\mathcal{D} \mapsto \mathcal{B}$ is $\epsilon$-differentially private ($\epsilon$-DP) if for every pair of neighboring inputs $D_1, D_2 \in \mathcal{D}$, and for every (measurable) output set $B \in \mathcal{B}$, the probabilities of events $M(D_1; F, \epsilon) \in B$ and $M(D_2; F, \epsilon) \in B$ are closer than a factor of $e^{\epsilon}$:
    \begin{equation}
        Pr(M(D_1; F, \epsilon) \in B) \leq e^{\epsilon} \cdot Pr(M(D_2; F, \epsilon) \in B).
    \label{eq:definition}
    \end{equation}
\end{definition}

In the context of differential privacy, consider a scenario in which there is a public-known privacy protection agreement that requires $\epsilon^{\prime}$ privacy budget. However, since more privacy budget (which means decreasing the privacy protection and making the results more distinguishable) often leads to better algorithm accuracy, the algorithm developer has the incentive to use some $\epsilon > \epsilon^{\prime}$ when performing the algorithm, which creates irresponsibility. Hence, we consider the state $\omega=g$ means $\epsilon = \epsilon^{\prime}$ and the state $\omega=b$ means $\epsilon > \epsilon^{\prime}$. Since privacy protection is often achieved by adding noise, it is assumed that for an algorithm $M$ with input dataset $D$, the privacy budget $\epsilon$  results in an output distribution $p(M(D)|\epsilon)$ for later usage. 

\begin{figure}
    \centering \vspace{-3mm}
    \includegraphics[width=3.45in]{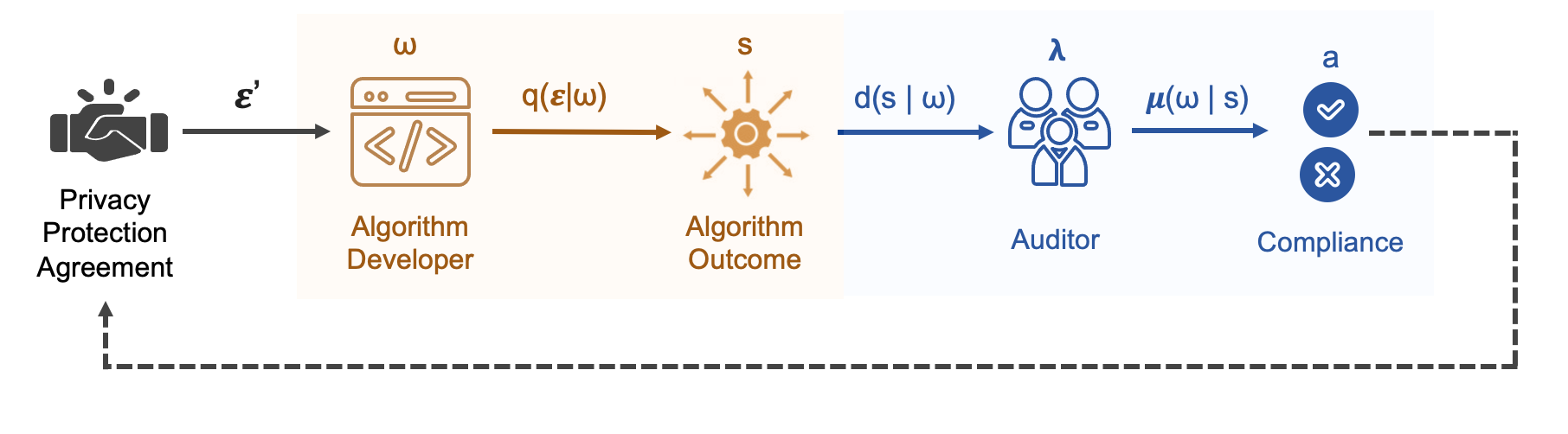}\vspace{-3mm}
    \caption{The Stackelberg herd audit game.}\vspace{-3mm}
    \label{fig:herd_audit}
\end{figure}

\subsection{Problem Setting for the Developer}

Consider two types of algorithm developers $g$ and $b$, and they play a mixed strategy for executing $\epsilon$, which are $q(\epsilon|g)$ and $q(\epsilon|b)$, respectively (for discrete choices of $\epsilon \in \mathcal{E}$). Each $\epsilon$ results in an algorithm accuracy $A(\epsilon)$, where $A: \mathcal{E} \mapsto \mathbb{R}$, under the assumption that a larger $\epsilon$ leads to better accuracy.

\begin{assumption}
    Given algorithm $M$ and input set $D$, a privacy budget $\epsilon$ leads to a unique output distribution $p(M(D)|\epsilon)$.
    \label{assump:output_distri}
\end{assumption}

\begin{assumption}
    For a given algorithm, the algorithmic accuracy under the privacy budget $\epsilon \in \mathcal{E} \subseteq \mathbb{R}$ is governed by $A: \mathcal{E} \mapsto \mathbb{R}$, and it is increasing in $\epsilon \in \mathcal{E}$.
    \label{assump:accu_increase}
\end{assumption}

In this context, the developer's strategy $q(\epsilon|\omega)$ given his type $\omega$ will lead to the distributions for the two hypotheses
\begin{align}
    Q_g(s) &= \sum\nolimits_{\epsilon}p(s|\epsilon) q(\epsilon|g), \label{eq:q_and_p_g}\\
    Q_b(s) &= \sum\nolimits_{\epsilon}p(s|\epsilon) q(\epsilon|b), \label{eq:q_and_p_b}
\end{align} 
where $p(s|\epsilon)$ is the distribution $p(M(D)|\epsilon)$ in Assumption \ref{assump:output_distri}.

\subsubsection{Responsible developer}
For a responsible algorithm developer, 
the mixed strategy $q(\epsilon|g)$ should have mass $0$ for $\epsilon > \epsilon^{\prime}$, which means that he always provides privacy protection at least complies with the agreement. Moreover, in order to maximize $A(\epsilon)$, a responsible algorithm developer tends to put all the mass on $\epsilon = \epsilon^{\prime}$ since $A(\epsilon) < A(\epsilon^{\prime}), \forall \epsilon < \epsilon^{\prime}$.

\begin{proposition}[Responsible Developer's Strategy]
    A responsible developer's mixed strategy reduces to a pure strategy by letting all the mass on $\epsilon = \epsilon^{\prime}$. Hence, $Q_g(s) = \sum_{\epsilon}p(s|\epsilon) q(\epsilon|g) = p(s|\epsilon^{\prime})$.
\end{proposition}

\subsubsection{Irresponsible developer}

However, it is important to consider various scenarios involving an irresponsible algorithm developer who prioritizes algorithm performance and disregards compliance with the agreement. If there is no auditor or no penalty imposed when the developer fails to pass the audit (i.e., when the auditor determines that $a = F$), the irresponsible developer can choose an extremely large value for $\epsilon$. Consequently, it is reasonable to assume that a penalty will be enforced if the irresponsible developer is detected. In such a situation, the irresponsible developer may attempt to maximize the probability of avoiding penalties, which is the probability of the auditor deciding $a = T$.

\begin{assumption}
    The irresponsible algorithm developer's mixed strategy will not put any mass on $\epsilon = \epsilon^{\prime}$. That is, $q(\epsilon^\prime|b)=0$.
\end{assumption}

\subsection{Revisiting the Auditor's Problem}

Considering that the penalty term for the irresponsible developer is influenced by the actions of the auditor, in terms of whether the irresponsible developer is caught or not, it is necessary to reexamine the problem from the auditor's perspective when the developer is also a strategic player aiming to evade the audit. We reformulate the auditor's problem by letting $u(g, T)=u(b, F)=0$ and setting the penalty terms $u(b, T)$ and $u(g, F)$ to negative values.
However, within the context of DP, it is important to note that the distributions for these hypotheses are predefined by the output distribution $p(s|\epsilon)$ and the developer's mixed strategy given the observed information $s$.
The audit confidences are analogous to those provided in  \eqref{eq:post_g} and \eqref{eq:post_b}.
\begin{assumption}
    Assume that $u(b, T) < 0$ and $u(g, F) < 0$ are the negative utilities for making wrong audit decisions.
\label{assump:penalty}
\end{assumption}

Given the distributions for the two hypotheses $Q_g(s)$ and $Q_b(s)$, the auditor aims to achieve the following:
\begin{equation}
\begin{aligned}
 \max _r \ \  &u(g, F) \sum_s\bigg[\sum_{\epsilon} p(s |\epsilon) q(\epsilon | g)\bigg] r(b | s) \\&+ u(b, T) \sum_s\bigg[\sum_{\epsilon} p(s |\epsilon) q(\epsilon | b)\bigg] r(g|s) \\&-\lambda \ \mathbb{E}_s\left[D_{k L}(r(\omega|s) \| \mu(\omega))\right], \\
 \text{s.t.}  \sum_{\omega}& r(\omega|s) = 1, r(w|s)\geq 0, \forall w\in\{g,b\}, \forall s \in \mathcal{S},
 \label{eq:auditor_prob}
\end{aligned}
\end{equation} where the first two terms put negative weights on the audit error, and the last term quantifies the expected reduction in uncertainty for the state $\omega$, measured in terms of the Kullback–Leibler (KL) divergence: $$\lambda \sum_s\bigg[Q_g(s) + Q_b(s)\bigg] \sum_\omega r(\omega| s) \log \frac{r(\omega|s)}{\mu(\omega)}.$$ 
%This term encapsulates the epistemic disparity and its influence on the informational cost by the parameter $\lambda$.
%
The decision of $r(g|s)$, $r(b|s)$ already incorporate the auditor's information strategy $d(s|\omega)$ since $r(\omega|s)=\frac{\mu(\omega)d(s|\omega)}{\sum_{\omega}\mu(\omega)d(s|\omega)}$. %(The observation follows from Remark \ref{remark:d_and_mu} in Section \ref{sec:choice_info}.)

%\begin{remark}
   % It might cause some confusion between the probability distribution of an output sample $s$ under both hypotheses $Q_g(s), Q_b(s)$, and the information obtaining strategy for the auditor $d(s|g), d(s|b)$.
%   It is important to distinguish between the probability distributions of an output sample s under both hypotheses, $Q_g(s), Q_b(s)$, and the information obtaining strategy for the auditor, $d(s|g), d(s|b)$. %This distinction helps avoid confusion between the two concepts.
%    Note that $Q_g(s), Q_b(s)$ are given by the output distribution based on assumption \ref{assump:output_distri} and the developer's mixed strategy $q(\epsilon|g), q(\epsilon|b)$, and that the information obtaining strategy on how to observe output sample $d(s|g), d(s|b)$ mainly comes from the auditor's preference over utility $u(\omega, a)$ and the epistemic factor $\lambda$, which leads to the audit confidence (posterior belief of the state) $r(g|s)$ and $r(b|s)$.
%\end{remark}

\subsection{Revisit the Irresponsible Developer's Problem}
Until now, the irresponsible developer's objective has become the following.
\begin{equation}
\begin{aligned}
    \max _{q(\cdot|b)}  &\sum_{\epsilon} q(\epsilon | b) A(\epsilon) + \beta  \sum_s\bigg[\sum_{\epsilon} p(s|\epsilon) q(\epsilon|b)\bigg] r(g |s),
\label{prob:bad_developer}
\end{aligned}
\end{equation} with $r(g |s)$ comes from the auditor's problem. The former term is the expected algorithm accuracy, and the latter term corresponds to the false negative rate of the auditor's decision, which is the rate of the irresponsible developer successfully passing the audit (and thus, the irresponsible developer seeks to maximize it). Note that $\beta >0$ indicates the irresponsible developer's preference for the two goals.

\section{Equilibrium Analysis}
For illustrative purposes, we work through an example where $|\mathcal{E}|=3$ in Appendix \ref{appendix: example}.
%For simplicity, we initially consider a scenario where the cardinality of the set $\mathcal{E}$ is three; i.e., $|\mathcal{E}|=3$ with $\mathcal{E}=\{\epsilon_l, \epsilon_m, \epsilon_h\}$, where $\epsilon_l < \epsilon_m < \epsilon_h$ and it's assumed that the claimed differential privacy budget is $\epsilon^{\prime} = \epsilon_l$.
%A more general case will be presented in Section \ref{sec:general}.
Besides, we assume that the distinguishability\textemdash quantified by distance measures such as the Kullback–Leibler divergence\textemdash between the output distributions $p(\cdot|\epsilon)$ and $p(\cdot|\epsilon^{\prime})$ increases when the difference between $\epsilon$ and $\epsilon^{\prime}$ expands.
%
%Then, $Q_g(s) = \sum_{\epsilon}p(s|\epsilon) q(\epsilon|g) = p(s|\epsilon^{\prime}) = p(s|\epsilon_l)$ and $Q_b(s) = \sum_{\epsilon}p(s|\epsilon) q(\epsilon|b) = p(s|\epsilon_m) q(\epsilon_m|b) + p(s|\epsilon_h) q(\epsilon_h|b)$. 

\subsection{The auditor's optimal strategy}
With the example in Appendix \ref{appendix: example} and derivations in Appendix \ref{appendix: r(w|s)}, the auditor's $r(g|s)$ and $r(b|s)$ that optimally solves problem \eqref{eq:auditor_prob} can be written as:
\begin{equation}
    r(g|s)=\mu(g)\exp \left(u(b, T)Q_b(s)/\lambda v(s)\right)/y^\prime(s),
    \label{eq:r_g}
\end{equation}
\begin{equation}
    r(b|s)=\mu(b)\exp \left(u(g, F)Q_g(s)/\lambda v(s)\right)/y^\prime(s),
\label{eq:r_b}
\end{equation}and $y^\prime(s)=\mu(g)\exp \big(\frac{u(b, T)Q_b(s)}{\lambda v(s)}\big)+\mu(b)\exp \big(\frac{u(g, F)Q_g(s)}{\lambda v(s)}\big)$ denotes the normalization term. We can observe that $r(g|s)$ and $r(b|s)$ share a similar form as \eqref{eq:post_g} and \eqref{eq:post_b}.

\begin{proposition}
The strategy specified by  \eqref{eq:r_g} and \eqref{eq:r_b} is optimal for the auditor with epistemic factor $\lambda$. 
\end{proposition}

\begin{remark}
    The results coincide with the intuition. We first take a look at $r(g|s)$. Recall that $\frac{u(b, T)}{\lambda}$ is negative. If the penalty term $u(b, T)$ is the same across all the auditors, the auditor with a larger epistemic factor $\lambda$ achieves $r(g|s)$ that is closer to $\mu(g)$. Combining with the auditor's objective in the maximization problem \eqref{eq:auditor_prob}, it means that the larger-$\lambda$ auditor might have a larger false negative error. Similarly, for $r(b|s)$, the larger-$\lambda$ auditor might have a larger false positive error.
\end{remark}

\subsection{The irresponsible developer's optimal strategy}
The irresponsible developer's problem \eqref{prob:bad_developer} is organized into 
\begin{equation}\label{eq:bad_developer_obj_2}
    \begin{aligned}
    \sum_{\epsilon} &q(\epsilon | b) A(\epsilon) + \beta \ \sum_s\left[\sum_{\epsilon} p(s|\epsilon) q(\epsilon|b)\right] r(g |s)\\
    &= \sum\limits_{\epsilon_{i}\in \mathcal{E}}q(\epsilon_i|b)\bigg[A(\epsilon_i) + \beta \sum_{s}r(g|s)p(s|\epsilon_i)\bigg].
\end{aligned}
\end{equation}
The irresponsible developer determines his optimal pure strategy $\epsilon$ to maximize (\ref{eq:bad_developer_obj_2}).
Specifically, the irresponsible developer assigns $q(\epsilon|b)=1$ to the $\epsilon$ that achieves the largest $\big[A(\epsilon) + \beta \sum_{s}r(g|s)p(s|\epsilon)\big]$. This leads us to the following propositions and remarks.
\begin{proposition}
    The irresponsible developer's optimal strategy is choosing the $\epsilon$ that maximizes $\big[A(\epsilon) + \beta \sum_{s}r(g|s)p(s|\epsilon)\big]$.
\end{proposition}

\begin{proposition}
    If algorithm accuracy $A(\epsilon)$ is increasing in $\epsilon$, the irresponsible developer  always chooses the largest $\epsilon$ if $r(g|s)=r(g|s^\prime), \forall s, s^\prime \in \mathcal{S}$.
\label{prop:same_s}
\end{proposition}
\begin{proof}
    We sketch the proof in Appendix \ref{appendix: same_s}.
\end{proof}

\begin{remark}
    The irresponsible developer violates as much as possible when the epistemic factor for the auditor $\lambda = \infty$.
\end{remark}

\begin{remark}
    When the auditor's epistemic factor $\lambda$ is small, indicating easy access to relevant information, an irresponsible developer is more likely to violate the claim.
\end{remark}

\begin{remark}
    %If the auditor's epistemic factor $\lambda$ is large, then the bad-developer who has a larger $\beta$ (values the successfully passing rate more) might also tend to violate the claim more.
    If the auditor's epistemic factor $\lambda$ is large, it is likely that an irresponsible developer with a larger $\beta$ (placing more value on the success rate of passing audits) will also tend to violate the claim more severely.
\end{remark}

\subsection{Auditor's audit confidence and epistemic factor}
With respect to Fig. \ref{fig:post_cost}, the optimal solution to the auditor's problem given by \eqref{eq:r_g} and \eqref{eq:r_b} establishes a relationship between the epistemic factor $\lambda$ and the auditor's confidence $r(\cdot|s)$ under fixed utilities $u(\omega, a)$.
Let $\chi(s)=[u(g, F)Q_g(s)-u(b, T) Q_b(s)]/v(s)$ and $\phi(s) =  [u(g, F)Q_g(s)+u(b, T) Q_b(s)]/(\lambda v(s))$ .
Taking the partial derivative of $r(g|s)$ with respect to $\lambda$ yields
$
        \frac{\partial r(g|s)}{\partial \lambda} 
        = \frac{\mu(g) \mu(b) \chi(s) \exp \left(\phi(s)\right)}{\lambda^2 y^\prime(s)^2}. 
$
Here, if the developer is irresponsible, then he never chooses a privacy budget $\epsilon$ that is equal to the claimed budget $\epsilon_{0}$.
Hence, $\chi(s)\neq0$.
The term $\partial r(g|s)/\partial \lambda$ is (strictly) positive if $\chi(s)>0$ and (strictly) negative otherwise.
When $\chi(s)<0$, $r(g|s)$ is close to $1$ when $\lambda$ goes close to $0$. 
Furthermore, the audit confidences for $g$ and $b$ become close to $0.5$ when $\lambda$ increases, which reveals that higher $\lambda$ leads to a weaker incentive to acquire more accurate information, thereby inducing lower audit confidences. 

Similarly, 
$
        \frac{\partial r(b|s)}{\partial \lambda} 
        = \frac{-\mu(g) \mu(b) \chi(s) \exp \left(\phi(s)\right)}{\lambda^2 y^\prime(s)^2}.
$ 
The term $\partial r(b|s)/\partial \lambda$ is positive if $-\chi(s)>0$ and negative otherwise.
When $-\chi(s)<0$, $r(b|s)$ is close to $1$ when $\lambda$ goes close to $0$. 
Furthermore, the audit confidences for $g$ and $b$ become closer to $0.5$ when $\lambda$ increases, which coincides with the setting that higher $\lambda$ leads to a weaker incentive to acquire more accurate information, thereby inducing lower audit confidences.
\begin{figure}[h]
\centering \vspace{-3mm}
    \begin{subfigure}[t]{0.24\textwidth}
        \centering
        \includegraphics[width=1.7in]{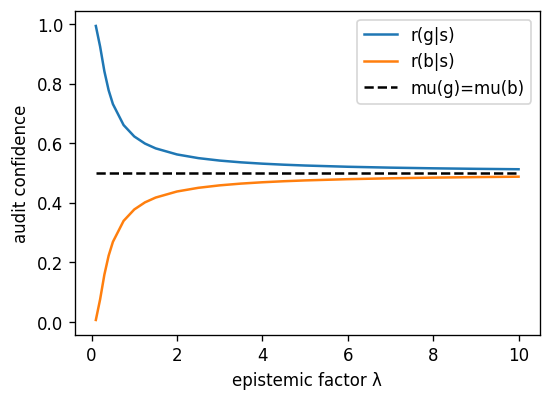}
        \caption{}
        \label{fig:post_and_cost}
    \end{subfigure}
    \begin{subfigure}[t]{0.24\textwidth}
        \centering
        \includegraphics[width=1.7in]{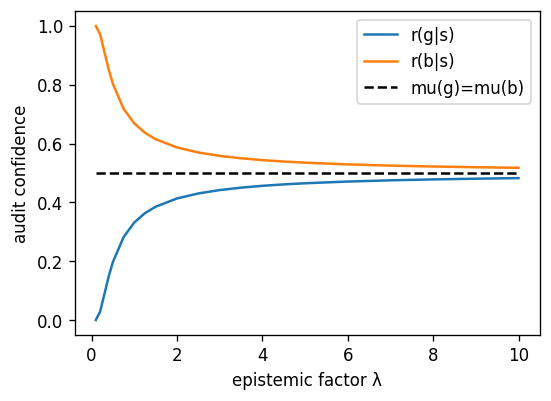}
        \caption{}
        \label{fig:post_cost_}
    \end{subfigure}
\caption{Trends between the auditor's confidence (posterior belief) and the auditor's epistemic factor ($\mu(g)=\mu(b)=0.5, u(b, T)=u(g, F)=-1$). Left: $\chi(s)<0$ ($Q_b(s)/v(s)=0.25$). Right: $-\chi(s)<0$ ($Q_b(s)/v(s)=0.85$).}
\vspace{-5mm}
\label{fig:post_cost}
\end{figure}
Note that audit confidence is determined by optimizing the objective, which consists of penalties for audit errors and costs associated with information acquisition. In this context, it is important to carefully select reasonable intervals for $u(\omega, a)$ and $\lambda$. In practice, as auditors are end-users for the algorithm, and given the disparities in end-users across different algorithms, the range for the epistemic factor needs to be contingent upon the ease with which corresponding end-users of the algorithm can access relevant information.

\subsection{Irresponsible developer's choice and auditor's confidence}
According to (\ref{eq:q_and_p_g}) and (\ref{eq:q_and_p_b}), the irresponsible developer's budget choice determines $Q_{b}(\cdot)$ given $p(\cdot)$.
Hence, \eqref{eq:r_g} and \eqref{eq:r_b} (shown in Fig. \ref{fig:post_and_Q}) also establish a relationship between the irresponsible developer's choice and the auditor's confidence.

\begin{figure}[ht]
    \centering \vspace{-3mm}
    \includegraphics[width=2in]{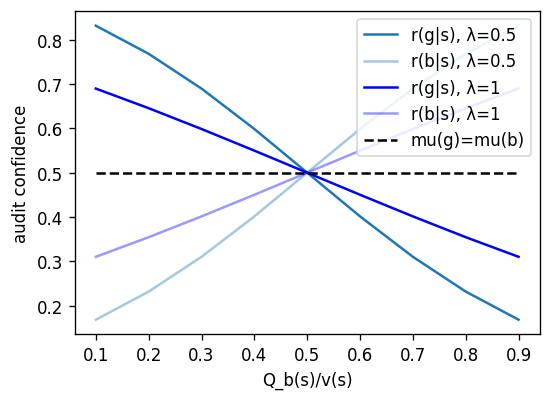}\vspace{-3mm}
    \caption{The trend between the auditor's confidence and the irresponsible developer's choice that leads to $Q_b(s)/v(s)$. Here, $\mu(g)=\mu(b)=0.5, u(b, T)=u(g, F)=-1, \lambda=1$.} \vspace{-3mm}
    \label{fig:post_and_Q}
\end{figure}

By taking partial derivative of $r(g|s)$ with respect to $\frac{Q_b(s)}{v(s)}$, we obtain 
$
        \frac{\partial r(g|s)}{\partial Q_b(s)/v(s)} 
        = \frac{\mu(g) \mu(b) (u(g, F)+u(b, T)) \exp \left(\phi(s)\right)}{\lambda y^\prime(s)^2}, 
$ which is negative since $u(g, F)+u(b, T) <0$. Additionally, as the value of $\lambda$ increases (when auditors incur higher costs for information acquisition), the magnitude of $\partial r(g|s)/\partial (\frac{Q_b(s)}{v(s)})$ decreases, implying relatively less influence on audit confidence. This trend is evident in Figure \ref{fig:post_and_Q}, where a greater $\lambda$ corresponds to a flatter curve for $r(g|s)$.

\section{Discussion and Conclusions}
Herd audit is a collective mechanism that empowers users to hold algorithm developers accountable, fostering the development of compliant and responsible digital products for the betterment of society. In this study, we examine herd audit through a game-theoretic lens, capturing the interactions between an idiosyncratic user and a privacy-preserving algorithm developer. Our framework adopts a Stackelberg game approach, enabling us to assess the impact of herd audit on responsible algorithm design and understand selfish and irresponsible strategies in worst-case scenarios.

We have specifically explored the presence of auditors with varying cognitive and reasoning capabilities, capturing epistemic disparities. Within our game-theoretic framework, we have consolidated the concept of rational inattention. The optimal strategy for auditors underscores the importance of easy access to relevant information, which enhances their confidence in the herd-audit process. Similarly, the optimal decision for algorithm developers has revealed that herd audit is a viable approach when auditors face lower costs in accessing knowledge, as denoted by smaller epistemic factors. Based on our findings, we conclude that herd audit poses a credit threat to developers and plays a vital role in promoting the responsible development of privacy-preserving algorithms. In future work, we aim to enrich the game-theoretic framework by incorporating end-users' incentives. This extension allows us to design an incentive mechanism that encourages participation in herd audits. Additionally, we plan to explore the fusion of distributed audits alongside a central audit center. Leveraging tools from decentralized hypothesis testing, game theory, information theory, and differential privacy, this research direction holds promise for advancing the field further.

%%%%%%%%%%%%%%%%%%%%%%%%%%%%%%%

\bibliography{reference}

\begin{thebibliography}{10}
\providecommand{\url}[1]{#1}
\csname url@samestyle\endcsname
\providecommand{\newblock}{\relax}
\providecommand{\bibinfo}[2]{#2}
\providecommand{\BIBentrySTDinterwordspacing}{\spaceskip=0pt\relax}
\providecommand{\BIBentryALTinterwordstretchfactor}{4}
\providecommand{\BIBentryALTinterwordspacing}{\spaceskip=\fontdimen2\font plus
\BIBentryALTinterwordstretchfactor\fontdimen3\font minus \fontdimen4\font\relax}
\providecommand{\BIBforeignlanguage}[2]{{%
\expandafter\ifx\csname l@#1\endcsname\relax
\typeout{** WARNING: IEEEtran.bst: No hyphenation pattern has been}%
\typeout{** loaded for the language `#1'. Using the pattern for}%
\typeout{** the default language instead.}%
\else
\language=\csname l@#1\endcsname
\fi
#2}}
\providecommand{\BIBdecl}{\relax}
\BIBdecl

\bibitem{Why_audit}
\BIBentryALTinterwordspacing
J.~Guszcza, I.~Rahwan, W.~Bible, M.~Cebrian, and V.~Katyal, ``Why we need to audit algorithms,'' 2018. [Online]. Available: \url{https://hdl.handle.net/21.11116/0000-0003-1C9E-D}
\BIBentrySTDinterwordspacing

\bibitem{DP_survey}
C.~Dwork, ``Differential privacy: A survey of results,'' in \emph{Theory and Applications of Models of Computation: 5th International Conference, TAMC 2008, Xi’an, China, April 25-29, 2008. Proceedings 5}.\hskip 1em plus 0.5em minus 0.4em\relax Springer, 2008, pp. 1--19.

\bibitem{VDP_detection}
Z.~Ding, Y.~Wang, G.~Wang, D.~Zhang, and D.~Kifer, ``Detecting violations of differential privacy,'' in \emph{Proceedings of the 2018 ACM SIGSAC Conference on Computer and Communications Security}, 2018, pp. 475--489.

\bibitem{dp_finder}
B.~Bichsel, T.~Gehr, D.~Drachsler-Cohen, P.~Tsankov, and M.~Vechev, ``Dp-finder: Finding differential privacy violations by sampling and optimization,'' in \emph{Proceedings of the 2018 ACM SIGSAC Conference on Computer and Communications Security}, ser. CCS '18.\hskip 1em plus 0.5em minus 0.4em\relax New York, NY, USA: Association for Computing Machinery, 2018, p. 508–524.

\bibitem{num_dp}
Y.~Han and S.~Mart{\'\i}nez, ``A numerical verification framework for differential privacy in estimation,'' \emph{IEEE Control Systems Letters}, vol.~6, pp. 1712--1717, 2021.

\bibitem{audit_survey}
J.~Bandy, ``Problematic machine behavior: A systematic literature review of algorithm audits,'' \emph{Proceedings of the acm on human-computer interaction}, vol.~5, no. CSCW1, pp. 1--34, 2021.

\bibitem{audit_transparent}
B.~Mittelstadt, ``Automation, algorithms, and politics| auditing for transparency in content personalization systems,'' \emph{International Journal of Communication}, vol.~10, p.~12, 2016.

\bibitem{collective_intelligence}
J.~M. Leimeister, ``Collective intelligence,'' \emph{Business \& Information Systems Engineering}, vol.~2, pp. 245--248, 2010.

\bibitem{eps_injustice}
M.~Fricker, \emph{Epistemic injustice: Power and the ethics of knowing}.\hskip 1em plus 0.5em minus 0.4em\relax Oxford University Press, 2007.

\bibitem{eps_science}
H.~Grasswick, ``Epistemic injustice in science,'' in \emph{The Routledge handbook of epistemic injustice}.\hskip 1em plus 0.5em minus 0.4em\relax Routledge, 2017, pp. 313--323.

\bibitem{manshaei2013game}
M.~H. Manshaei, Q.~Zhu, T.~Alpcan, T.~Bac{\c{s}}ar, and J.-P. Hubaux, ``Game theory meets network security and privacy,'' \emph{ACM Computing Surveys (CSUR)}, vol.~45, no.~3, pp. 1--39, 2013.

\bibitem{fang2021introduction}
F.~Fang, S.~Liu, A.~Basak, Q.~Zhu, C.~D. Kiekintveld, and C.~A. Kamhoua, ``Introduction to game theory,'' \emph{Game Theory and Machine Learning for Cyber Security}, pp. 21--46, 2021.

\bibitem{RI_discrete}
F.~Mat{\v{e}}jka and A.~McKay, ``Rational inattention to discrete choices: A new foundation for the multinomial logit model,'' \emph{American Economic Review}, vol. 105, no.~1, pp. 272--298, 2015.

\bibitem{RI_info_cost}
A.~Caplin and M.~Dean, ``Revealed preference, rational inattention, and costly information acquisition,'' \emph{American Economic Review}, vol. 105, no.~7, pp. 2183--2203, July 2015.

\bibitem{crowdsensing}
F.~Restuccia, N.~Ghosh, S.~Bhattacharjee, S.~K. Das, and T.~Melodia, ``Quality of information in mobile crowdsensing: Survey and research challenges,'' \emph{ACM Transactions on Sensor Networks (TOSN)}, vol.~13, no.~4, pp. 1--43, 2017.

\bibitem{crowdsourcing}
Y.~Zhao and Q.~Zhu, ``Evaluation on crowdsourcing research: Current status and future direction,'' \emph{Information systems frontiers}, vol.~16, pp. 417--434, 2014.

\bibitem{crowd_defence}
J.~Pawlick and Q.~Zhu, ``Active crowd defense,'' \emph{Game Theory for Cyber Deception: From Theory to Applications}, pp. 147--167, 2021.

\bibitem{frye2021technology}
H.~Frye, ``The technology of public shaming,'' \emph{Social Philosophy and Policy}, vol.~38, no.~2, pp. 128--145, 2021.

\bibitem{yu2016mitigating}
H.~Yu, C.~Miao, C.~Leung, Y.~Chen, S.~Fauvel, V.~R. Lesser, and Q.~Yang, ``Mitigating herding in hierarchical crowdsourcing networks,'' \emph{Scientific reports}, vol.~6, no.~1, p.~4, 2016.

\bibitem{comeig2020rational}
I.~Comeig, E.~Mesa-V{\'a}zquez, P.~Sendra-Pons, and A.~Urbano, ``Rational herding in reward-based crowdfunding: An mturk experiment,'' \emph{Sustainability}, vol.~12, no.~23, p. 9827, 2020.

\bibitem{eickhoff2018cognitive}
C.~Eickhoff, ``Cognitive biases in crowdsourcing,'' in \emph{Proceedings of the eleventh ACM international conference on web search and data mining}, 2018, pp. 162--170.

\bibitem{morris2012priming}
R.~R. Morris, M.~Dontcheva, and E.~M. Gerber, ``Priming for better performance in microtask crowdsourcing environments,'' \emph{IEEE Internet Computing}, vol.~16, no.~5, pp. 13--19, 2012.

\bibitem{karger2014budget}
D.~R. Karger, S.~Oh, and D.~Shah, ``Budget-optimal task allocation for reliable crowdsourcing systems,'' \emph{Operations Research}, vol.~62, no.~1, pp. 1--24, 2014.

\bibitem{wang2016toward}
K.~Wang, X.~Qi, L.~Shu, D.-j. Deng, and J.~J. Rodrigues, ``Toward trustworthy crowdsourcing in the social internet of things,'' \emph{IEEE Wireless Communications}, vol.~23, no.~5, pp. 30--36, 2016.

\bibitem{allahbakhsh2012reputation}
M.~Allahbakhsh, A.~Ignjatovic, B.~Benatallah, E.~Bertino, N.~Foo \emph{et~al.}, ``Reputation management in crowdsourcing systems,'' in \emph{8th International conference on collaborative computing: networking, applications and worksharing (CollaborateCom)}.\hskip 1em plus 0.5em minus 0.4em\relax IEEE, 2012, pp. 664--671.

\bibitem{yu2020crowdr}
Y.~Yu, S.~Liu, L.~Guo, P.~L. Yeoh, B.~Vucetic, and Y.~Li, ``Crowdr-fbc: A distributed fog-blockchains for mobile crowdsourcing reputation management,'' \emph{IEEE Internet of Things Journal}, vol.~7, no.~9, pp. 8722--8735, 2020.

\bibitem{Cognitive-behavioral}
A.~A. Gonz{\'a}lez-Prendes and S.~M. Resko, ``Cognitive-behavioral theory,'' 2012.

\bibitem{why_cognitive}
D.~Fum, F.~Del~Missier, A.~Stocco \emph{et~al.}, ``The cognitive modeling of human behavior: Why a model is (sometimes) better than 10,000 words,'' \emph{Cognitive Systems Research}, vol.~8, no.~3, pp. 135--142, 2007.

\bibitem{huang2023cognitive}
L.~Huang and Q.~Zhu, \emph{Cognitive Security: A System-Scientific Approach}.\hskip 1em plus 0.5em minus 0.4em\relax Springer Nature, 2023.

\bibitem{anderson2004integrated}
J.~R. Anderson, D.~Bothell, M.~D. Byrne, S.~Douglass, C.~Lebiere, and Y.~Qin, ``An integrated theory of the mind.'' \emph{Psychological review}, vol. 111, no.~4, p. 1036, 2004.

\bibitem{Sims_2003}
C.~A. Sims, ``Implications of rational inattention,'' \emph{Journal of monetary Economics}, vol.~50, no.~3, pp. 665--690, 2003.

\bibitem{rajtmajer2017ultimatum}
S.~Rajtmajer, A.~Squicciarini, J.~M. Such, J.~Semonsen, and A.~Belmonte, ``An ultimatum game model for the evolution of privacy in jointly managed content,'' in \emph{Decision and Game Theory for Security: 8th International Conference, GameSec 2017, Vienna, Austria, October 23-25, 2017, Proceedings}.\hskip 1em plus 0.5em minus 0.4em\relax Springer, 2017, pp. 112--130.

\bibitem{casorran2019study}
C.~Casorr{\'a}n, B.~Fortz, M.~Labb{\'e}, and F.~Ord{\'o}{\~n}ez, ``A study of general and security stackelberg game formulations,'' \emph{European journal of operational research}, vol. 278, no.~3, pp. 855--868, 2019.

\bibitem{guerrero2018solving}
D.~Guerrero, A.~A. Carsteanu, and J.~B. Clempner, ``Solving stackelberg security markov games employing the bargaining nash approach: Convergence analysis,'' \emph{Computers \& Security}, vol.~74, pp. 240--257, 2018.

\bibitem{chen2016optimal}
J.~Chen and Q.~Zhu, ``Optimal contract design under asymmetric information for cloud-enabled internet of controlled things,'' in \emph{Decision and Game Theory for Security: 7th International Conference, GameSec 2016, New York, NY, USA, November 2-4, 2016, Proceedings}.\hskip 1em plus 0.5em minus 0.4em\relax Springer, 2016, pp. 329--348.

\bibitem{zhang2019mathtt}
R.~Zhang and Q.~Zhu, ``Flipin: A game-theoretic cyber insurance framework for incentive-compatible cyber risk management of internet of things,'' \emph{IEEE Transactions on Information Forensics and Security}, vol.~15, pp. 2026--2041, 2019.

\bibitem{zhu2012guidex}
Q.~Zhu, C.~Fung, R.~Boutaba, and T.~Basar, ``Guidex: A game-theoretic incentive-based mechanism for intrusion detection networks,'' \emph{IEEE Journal on Selected Areas in Communications}, vol.~30, no.~11, pp. 2220--2230, 2012.

\bibitem{huang2021duplicity}
L.~Huang and Q.~Zhu, ``Duplicity games for deception design with an application to insider threat mitigation,'' \emph{IEEE Transactions on Information Forensics and Security}, vol.~16, pp. 4843--4856, 2021.

\bibitem{deception}
K.~Hor{\'a}k, Q.~Zhu, and B.~Bo{\v{s}}ansk{\`y}, ``Manipulating adversary’s belief: A dynamic game approach to deception by design for proactive network security,'' in \emph{Decision and Game Theory for Security: 8th International Conference, GameSec 2017, Vienna, Austria, October 23-25, 2017, Proceedings}.\hskip 1em plus 0.5em minus 0.4em\relax Springer, 2017, pp. 273--294.

\bibitem{evasion_detectors}
Y.~Hu and Q.~Zhu, ``Evasion-aware neyman-pearson detectors: A game-theoretic approach,'' in \emph{2022 IEEE 61st Conference on Decision and Control (CDC)}, 2022, pp. 6111--6117.

\bibitem{early_detection}
S.~N. Narayanan, A.~Ganesan, K.~Joshi, T.~Oates, A.~Joshi, and T.~Finin, ``Early detection of cybersecurity threats using collaborative cognition,'' in \emph{2018 IEEE 4th International Conference on Collaboration and Internet Computing (CIC)}, 2018, pp. 354--363.

\end{thebibliography}
\bibliographystyle{IEEEtran}

%%%%%%%%%%%%%%%%%%%%%%%%%%%%%%%
\appendix
\subsection{Proof of optimal decision rule $\delta^*(s)$}\label{appendix: delta(s)}
For the null hypothesis, $d(s|g)$, and the alternative hypothesis, $d(s|b)$, the expected utility for problem $HT$ in \eqref{eq:4} can be reformulated as
    \begin{equation*}
    \begin{aligned}
    \mathbb{E}[u(\omega, a)]
    &= \mu(g)u(g, F) + \mu(b)u(b, F) \\
    &+ \sum_{\delta(s)=T} \bigg\{\mu(g)\big[u(g, T)-u(g, F)\big]d(s|g)\\
    &-\mu(b)\big[u(b, F)-u(b, T)\big]d(s|b)\bigg\}.
    \end{aligned}
    \end{equation*} Therefore, to maximize the expected utility, the auditor must decide $\delta(s)=T$ if $\mu(g)\big[u(g, T)-u(g, F)\big]d(s|g)>\mu(b)\big[u(b, F)-u(b, T)\big]d(s|b)$. This completes the proof.

\subsection{Proof of optimal information strategy $d$}\label{appendix: d(s|w)}
To analyze the problem, we use the method of Lagrange multipliers and denote 
\begin{align*}
    &J(d, y) = \sum_{\omega}\sum_{a}\mu(\omega)u(\omega, a)\sum_{s:\delta^*_d(s)=a}d(s|\omega)\\
    &- \lambda \sum_{\omega}\sum_{s} d(s| \omega) \mu(\omega) \ln \frac{d(s | \omega)}{\sum_{\omega} d(s | \omega) \mu(\omega)}
    -\sum_{\omega}y(\omega)d(s|\omega),
\end{align*} with the last term corresponding to the constraint that $d(s|\omega)$ should be a conditional probability mass function.

Then, for $d(s|g)$ with $s \in S_{d, T}$, according to the first-order and the second-order condition,
\begin{align*}
    &\frac{\partial J(d, y)}{\partial d(s|g)} = \mu(g)u(g, T)-\lambda \mu(g) \log(\frac{d(s|g)}{v(s)})-y(g)=0,\\
    &\frac{\partial^2 J(d, y)}{\partial d(s|g)^2} = -\lambda \mu(g) \frac{\mu(b)d(s|b)}{d(s|g)v(s)} \leq 0, 
\end{align*} where $v(s)=\sum_{\omega}d(s|\omega)\mu(\omega)$. Letting $\log(y'(g))=\frac{y(g)}{\lambda \mu(g)}$ leads to the following $d(s|g)$ that maximizes \eqref{eq:max_d}.
\begin{align*}
    &\lambda \mu(g) \bigg[\frac{u(g, T)}{\lambda}-\log(\frac{d(s|g)}{v(s)})-\log y'(g)\bigg]=0,\\
    &d(s|g) = \frac{v(s)\exp(\frac{u(g, T)}{\lambda})}{y'(g)}.
\end{align*}Note that $y'(g)$ is the normalization term. Similarly, for $d(s|g)$ with $s \in S_{d, F}$ and $d(s|b)$, we can get the information-obtaining strategy in \eqref{eq:d_g} and \eqref{eq:d_b}

\subsection{An illustrative example for equilibrium analysis} \label{appendix: example}
We consider a scenario where the cardinality of the set $\mathcal{E}$ is three; i.e., $|\mathcal{E}|=3$ with $\mathcal{E}=\{\epsilon_l, \epsilon_m, \epsilon_h\}$, where $\epsilon_l < \epsilon_m < \epsilon_h$ and it's assumed that the claimed differential privacy budget is $\epsilon^{\prime} = \epsilon_l$.
Then, the two hypotheses become $$Q_g(s) = \sum_{\epsilon}p(s|\epsilon) q(\epsilon|g) = p(s|\epsilon^{\prime}) = p(s|\epsilon_l),$$ $$Q_b(s) = \sum_{\epsilon}p(s|\epsilon) q(\epsilon|b) = p(s|\epsilon_m) q(\epsilon_m|b) + p(s|\epsilon_h) q(\epsilon_h|b).$$ 

According to derivations in Appendix \ref{appendix: r(w|s)}, the strategy specified by \eqref{eq:r_g} and \eqref{eq:r_b} is optimal for the auditor with epistemic factor $\lambda$. 

We then shift our focus to the irresponsible developer's strategy. The irresponsible developer endeavors to enhance algorithmic accuracy while concurrently maximizing the probability of evading detection by the auditor, thereby increasing the likelihood of being perceived as a responsible developer.
Hence, the irresponsible developer's decision-making can be described by the following optimization problem:
\begin{equation}\label{eq:bad_developer_prob_1}
    \begin{aligned}
        \max _{q(\cdot|b)} \ \ &\big[q(\epsilon_m | b) A(\epsilon_m) + q(\epsilon_h | b) A(\epsilon_h)\big]\\
        &+\beta \ \sum_s\bigg[p(s|\epsilon_m)q(\epsilon_m|b)+p(s|\epsilon_h)q(\epsilon_h|b)\bigg] r(g |s).
    \end{aligned}
\end{equation} 
By leveraging $q(\epsilon_m|b)=1-q(\epsilon_h|b)$, we rewrite the problem \eqref{eq:bad_developer_prob_1} as follows:
\begin{equation}
    \begin{aligned}
        \max _{q(\epsilon_l|b)} \ \ &A(\epsilon_h) + \beta \ \sum_s r(g|s)p(s|\epsilon_h) +\bigg\{\big[A(\epsilon_m)-A(\epsilon_h)\big]\\&+\beta \sum_s r(g|s) \big[p(s|\epsilon_m)-p(s|\epsilon_h)\big]\bigg\}q(\epsilon_m|b).
    \end{aligned}
\label{eq:owner_decision}
\end{equation} 
Since the first two terms $A(\epsilon_h) + \beta \ \sum_s r(g|s)p(s|\epsilon_h)$ are independent of $q(\cdot|b)$, \eqref{eq:owner_decision} suggests the following strategy for the irresponsible developer: let $\Delta A=A(\epsilon_m)-A(\epsilon_h)$,
\[
\begin{cases}
    q(\epsilon_m|b)=1, & \Delta A + \beta \sum_s r(g|s) \big[p(s|\epsilon_m)-p(s|\epsilon_h)\big] > 0\\
    q(\epsilon_h|b)=1, & \textup{otherwise.}
\end{cases}
\]
That is, the irresponsible developer has a pure strategy by choosing either $q(\epsilon_m|b)=1$ or $q(\epsilon_h|b)=1$.

\subsection{Proof of auditor's strategy $r$}\label{appendix: r(w|s)}
In (\ref{eq:auditor_prob}), the KL divergence term with a negative sign is concave with respect to the decision variables $r(\cdot)$ given fixed priors $\mu(\cdot)$.
Therefore, the combination of the terms in the objective function forms a weighted sum of concave functions. This makes the overall objective function concave.
Given the linear constraints, the feasibility set is convex. Hence, the optimization problem (\ref{eq:auditor_prob}) is a concave maximization over a convex set.

The Lagrangian corresponding to (\ref{eq:auditor_prob}) is then given by  
\begin{equation}
\begin{aligned}
    J(r, y, z) = &\ u(g, F) \sum_s\bigg[\sum_{\epsilon} p(s |\epsilon) q(\epsilon | g)\bigg] r(b | s) \\
 &+ u(b, T) \sum_s\bigg[\sum_{\epsilon} p(s |\epsilon) q(\epsilon | b)\bigg] r(g|s) \\
 &-\lambda \ \sum_s\bigg[\sum_{\epsilon} p(s |\epsilon) q(\epsilon |g)+\sum_{\epsilon} p(s|\varepsilon) q(\epsilon |b)\bigg] \\
 &\sum_\omega r(\omega| s) \log \frac{r(\omega|s)}{\mu(\omega)}-y r(g|s) -z r(b|s),
\end{aligned}
\end{equation}
where $y \in \mathbb{R}$ and $z \in \mathbb{R}$ are the associated Lagrange multipliers.
Then, the first-order condition concerning $r(g|s)$ implies
\begin{equation*}
\begin{aligned}
    \frac{\partial J(r)}{\partial r(g|s)} 
    &= u(b, T) \bigg[p(s|\epsilon_m)q(\epsilon_m|b)+p(s|\epsilon_h)q(\epsilon_h|b)\bigg]\\ 
    &-\lambda \bigg[p(s|\epsilon_l)+ p(s|\epsilon_m)q(\epsilon_m|b)+p(s|\epsilon_h)q(\epsilon_h|b) \bigg] \\
    &\bigg(\log \frac{r(g|s)}{\mu(g)}+1 \bigg)-y\\
    &=0.
\end{aligned}
\end{equation*}
Hence, we obtain
\begin{equation*}
\begin{aligned}
     u(b,T)Q_b(s)-\lambda v(s)\bigg(\log \frac{r(g|s)}{\mu(g)}+1\bigg)-y=0,\\
     \frac{u(b, T)Q_b(s)}{\lambda v(s)}-\bigg(\frac{y}{\lambda v(s)}+1\bigg) = \log \frac{r(g|s)}{\mu(g)}.
\end{aligned}
\end{equation*} By letting $\log y'(s)=\big(\frac{y}{\lambda v(s)}+1\big)$, $r(g|s)$ can then be written as \eqref{eq:r_g}. We can get $r(b|s)$ described in \eqref{eq:r_b} with a similar process. 

\subsection{Proof of Proposition \ref{prop:same_s}}\label{appendix: same_s}

We sketch the proof for $|\mathcal{S}|=2$ with $\mathcal{S}=\{s_1, s_2\}$. In this example, $p(s_1|\epsilon_m)+p(s_2|\epsilon_m)=1$ and $p(s_1|\epsilon_h)+p(s_2|\epsilon_h)=1$, then $r(g|s_1)[p(s_1|\epsilon_m)-p(s_1|\epsilon_h)] + r(g|s_2)[p(s_2|\epsilon_m)-p(s_2|\epsilon_h)] = 1-1=0$ if $r(g|s_1)=r(g|s_2)$. 

Hence, $A(\epsilon_m)-A(\epsilon_h) +\sum_s r(g|s) \big[p(s|\epsilon_m)-p(s|\epsilon_h)\big] < 0$ in the case where $A(\epsilon_m)< A(\epsilon_h)$, which leads to $q(\epsilon_m|b)=0$ and $q(\epsilon_h|b)=1$. This completes the proof.

\end{document}